\pdfoutput=1

\documentclass[oneside,twocolumn]{article}

\usepackage[T1]{fontenc} 
\usepackage{microtype} 

\usepackage[english]{babel} 

\usepackage[hmarginratio=1:1,top=32mm,columnsep=20pt,textwidth=16cm]{geometry} 
\usepackage[hang, small,labelfont=bf,up,textfont=it,up]{caption} 
\usepackage{booktabs} 

\usepackage{enumitem} 
\setlist[itemize]{noitemsep} 

\usepackage{abstract} 

\usepackage{titlesec} 
\renewcommand\thesubsection{\roman{subsection}} 
\titleformat{\section}[block]{\large\scshape\centering}{\thesection.}{1em}{} 
\titleformat{\subsection}[block]{\large}{\thesubsection.}{1em}{} 

\newcommand{\bingbat}{\vspace{-.8em}\par\noindent\begin{center}\rule{0.25\columnwidth}{0.4pt}\end{center}}

\usepackage{fancyhdr} 
\usepackage{titling} 

\usepackage{hyperref} 

\usepackage{amsmath}
\usepackage{amssymb}
\usepackage{soul}
\usepackage{physics}
\usepackage{csquotes}
\usepackage{graphicx}
\usepackage{xfrac}
\usepackage{bm}
\usepackage[b]{esvect}
\usepackage{mathtools}

\usepackage{amsthm}
\usepackage{cleveref}

\usepackage{newpxtext} 
\usepackage[varbb]{newpxmath}

\theoremstyle{definition}
\newtheorem*{definition}{Definition}
\theoremstyle{definition}
\newtheorem*{corollary}{Observation}
\theoremstyle{definition}
\newtheorem{thm}{Theorem}

\newcommand{\Ry}{\textsc{Ry}}
\newcommand{\Rz}{\textsc{Rz}}
\newcommand{\X}{\textsc{X}}
\newcommand{\Cx}{\textsc{CX}}
\newcommand{\Ccx}{\ensuremath{\textsc{C}^2\X}}

\newcommand{\Ctrln}[1]{\ensuremath{\textsc{C}_n#1}}
\renewcommand{\vec}[1]{\vv{\bm{#1}}}

\usepackage{titletoc}

\usepackage{array}
\newcolumntype{R}[1]{>{\raggedleft\let\newline\\\arraybackslash\hspace{0pt}}p{#1}}

    \titlecontents{chapter}
    [0em] %
    {\medskip\bfseries}
    {\thecontentslabel\quad}
    {}
    {\hfill\contentspage}

    \titlecontents{section}[1.72em]{\smallskip}%
    {\thecontentslabel.\enspace}
    {}
    {\titlerule*[1.2pc]{.}\contentspage}%

\makeatletter
                {\list{}{\leftmargin=0pt 
                        \labelwidth\z@ \itemindent-\leftmargin
                        }}%
                {\endlist}
\makeatother

\usepackage[
  backend=biber,
  style=numeric-comp,
  sorting=none,
  url=false,
  isbn=false
]{biblatex}
\pretitle{\begin{center}\Huge\bfseries} 
\posttitle{\end{center}} 
\title{Quantum Circuits\\in Additive Hilbert Space} 
\makeatletter \newcommand{\mytitle}{\@title} \makeatother
\author{%
\textsc{Luca Mondada}\\[1ex]
\normalsize Department of Computer Science, University of Oxford \\
\normalsize \href{mailto:luca.mondada@cs.ox.ac.uk}{luca.mondada@cs.ox.ac.uk}
}
\date{25th October 2021} 
\renewcommand{\maketitlehookd}{%
\begin{abstract}
\noindent %
Representations of quantum computations are almost always based on a tensor
product $\otimes$-structure.
This coincides with what we are able to execute in our experiments,
as well as what we observe in Nature,
but it makes certain familiar quantum primitives convoluted.
Reversible classical circuits, diagonal operations or controlled unitaries
all have very elegant and simple matrix representations that cannot be
expressed succinctly as a circuit in a simple gate set,
complicating quantum algorithm design and circuit optimization.

We propose a new additive presentation of quantum computation to address this.
We show how conventional circuits can be expressed in the additive space
and how they can be recovered.
In particular in our formalism we are able to synthesize high-level multi-controlled primitives
from low-level circuit decompositions, making it an invaluable
tool for circuit optimization.
Our formulation also accepts a circuit-like diagrammatic representation 
and proposes a novel and simple interpretation of quantum computation.

\end{abstract}
}

\begin{document}

\maketitle

\setcounter{tocdepth}{1}
\setcounter{secnumdepth}{2}
{\tableofcontents}


\section{Circuits and Matrices are~not~Enough}

The quantum circuit, as an abstraction for quantum computation, has been widely adopted in quantum
computing~\cite{Nielsen2002}.
It can represent arbitrary unitary maps and be readily transpiled to hardware-level instructions;
it has thus become the \emph{de facto} standard for the design and execution of quantum computations~\cite{Openqasm}.
However, circuits also impose constraints on both ends of the quantum computing stack.

The viability of early quantum applications
will rest on execution optimization and noise mitigation techniques that must be tailored
and adjusted to characteristics of the specific hardware~\cite{Preskill2018}.
Emerging architectures such as ion traps~\cite{Schindler2013,Harty2014,Grzesiak2020,Figgatt2019},
cold atoms~\cite{Saffman2010,Scholl2021} and photonics~\cite{bartolucci2021,Arrazola2021} are widely
expected to introduce
unique hardware capabilities that may for instance
differ from the familiar two-qubit entangling primitives~\cite{Grzesiak2020,Figgatt2019,Martinez2016}.
The limited expressivity of the chosen gate set inherent to the circuit model 
severely limits the use of such architecture-specific instructions.

At the other end of the spectrum, circuit constrains quantum algorithm design.
In the circuit model, individual gates, typically defined in terms of their matrices,
are presented as arbitrary building blocks that the user is expected to compose into 
meaningful (and -- please -- fast) algorithms.
Unfortunately, most of us have come to the realisation
that this is a highly challenging task, for all but a few applications
that themselves present clear ``quantum-like''
structure~\cite{Aspuru2005,Lloyd1996,Meichanetzidis2021,Kartsaklis2021}.

Various other representations have been explored to address the disadvantages of quantum
circuits, whether it be for better circuit optimization and
execution performance~\cite{Stojkovic2020,Amy2019,Cowtan2020,Bae2020}
or for design and educational purposes~\cite{Abdollahi2006,coecke2018}.
By far the most common choice, however, is the matrix representation of the unitary computation.

Quantum computations, and hence quantum circuits, can always be expressed as linear unitary
transformation of finite dimensional Hilbert spaces, and as such,
can always be described by a matrix.

The matrix representation has multiple advantages that
can address some of the drawbacks of circuits discussed previously.
It is first of all 
a straightforward way to introduce quantum computing, as most audiences will be familiar
with the limited linear algebra that is required.
This simplicity also helps in algorithm design, as hardware details are abstracted away
to leave only algebraic considerations.

Matrices are furthermore unique, making them a useful canonical representation in many
circumstances~\cite{Amy2013818,Hales2000,Shende2002,Li2014}.
For instance, the most successful circuit optimization techniques to date
are in fact not circuit optimization:
they are circuit synthesis approaches that reduce an input circuit to its unique matrix representation,
which can in turn be decomposed optimally into a shorter circuit~\cite{Davis2020,Smith2021,Gheorghiu2021,Wu2020}.

It is also much simpler to obtain usual fidelity metrics for quantum computations such
as the Hilbert Schmidt from matrix representations~\cite{Jozsa1994}.
As a result, approximate compilation approaches, where precision in the computation is traded in
for lower noise, struggle to use circuit-based internal representations,
choosing to rely on matrix semantics instead~\cite{Madden2021}.

\bingbat

There is a significant conceptual gap between the matrix and the circuit representation.
The tensor product underlying quantum circuits in particular is often glossed over
when introducing the notation.
However, with this structure comes a counter-intuitive exponential scaling of the state space dimension,
which is unnatural to represent using component-wise vector and matrix representations.
This makes the link between linear algebra computations performed on the state space
and circuit representation hard to grasp and internalize.

In this paper, we propose a new representation of quantum computation in which the state space
is considered with its additive direct sum $\oplus$-structure instead of the tensor product $\otimes$
common in most quantum mechanics formulations.
It is in this sense a close parent of the matrix representation.

As a consequence, our new formulation lies between
the two representations just discussed: it combines the transparency of the linear algebra formalism
with the convenience of a diagrammatic language close to hardware primitives.
In this framework, the matrix of the computation can easily be visualized from the circuit-like presentation,
while executable circuits can still be synthesised.

We contend that this new formulation of quantum computation
provides a novel, simple model of quantum computation that bridges the gap between existing representations.
We believe this can support educational initiatives and quantum algorithm design,
as well as conceive new quantum circuit optimization that would improve quantum computation execution on hardware.

\subsubsection{Related work}
Quantum computer science and quantum mechanics in general have a rich history in diagrammatic
representations~\cite{Feynman1949,Coecke2008,coecke2018,vandewetering2020}.
Our approach draws a lot in spirit to diagrammatic formalizations of quantum computation
such as the ZX calculus and its variations~\cite{vandewetering2020,Backens2019,Roy2011,Wang2019}.
Just like the circuit model and unlike our proposal,
these calculi all rely on the multiplicative tensor $\otimes$-structure.

Diagrammatic theories of linear algebra with an additive $\oplus$-structure, termed signal flow graphs,
have also been proposed previously by Bonchi, Soboci{\'{n}}ski and Zanasi~\cite{Bonchi2015,Bonchi2017,Bonchi2021}.
They developed a universal and complete representation of finite dimensional linear algebra.

None of the diagrammatic calculi above, however, propose a framework for unitary-only operations,
a situation that makes circuit extraction from such representations difficult~\cite{Backens2020}.
In this sense, the diagrammatic primitives in our proposal have semantics that are closer to quantum circuits.

Our work also draws heavily on various optimization techniques, and in particular on circuit synthesis algorithms
for classical reversible circuits and unitary matrices.

Classical reversible circuit optimizations have a long history~\cite{Shende2002,Susam20171,Li2014,Stojkovic2020,Scott2008231,Arabzadeh2010849,Datta2013209,Datta2013316,Gado2021}.
We highlight in particular a strategy by Stojković and colleagues that can trade qubit ancillas for
reduced circuit depth using binary decision diagrams~\cite{Stojkovic2020} as well as Susam and Altun's sorting-based
circuit synthesis strategy~\cite{Susam20171}.

Unitary decomposition is another successful field for circuit optimization~%
\cite{Iten2019,Davis2020,Smith2021,Gheorghiu2021,Wu2020,Krol2021,Loke2014}.
Aside from state-of-the-art performance for search-based circuit optimization that we touched on earlier,
we mention Loke, Wang and Chen's circuit synthesis approach that
combines traditional matrix decomposition techniques
with reversible classical circuit optimization to outperform matrix decomposition alone.
We will see that such approaches are particularly suited to the additive diagrams that we are proposing.

\subsubsection{Notation}
Prepending a quantum gate with $\textsc{C}^n$ refers to the gate
controlled on $n$-qubits: e.g. $\textsc{C}^n\X$. If $n=1$, we usually omit the superscript.
We use subscript indices on gates to indicate which qubit or dimension it is applied to: $\Ry_{2}(\theta)$
is a rotation applied to the second qubit.
We represent bitstrings using the Little Endian convention: given three qubits 
\texttt{q0}, \texttt{q1} and \texttt{q2}, a bitstring $\ket{cba}$ means
(\texttt{q0}: $\ket{a}$, \texttt{q1}: $\ket{b}$ \texttt{q2}: $\ket{c}$).

When using $k$, $n$ or $m$ in relation to the number of qubits or dimensions of a system, it is always implicitly
assumed that $k,n,m \in \mathbb{N}^+$.
Similarly, rotation angles $\theta,\psi$ are always real numbers $\theta,\psi \in \mathbb{R}$.
We use blackboard boldface $\mathbb{0}_n, \mathbb{1}_n$ to refer to $n\times n$-dimensional zero
and identity matrices, respectively.
We use $\mathcal{H}_n$ to denote an $n$-dimensional Hilbert space.

The rest of our notation is standard in linear algebra.
$\mathcal{H}_n \oplus \mathcal{H}_m$ denotes the direct sum of Hilbert spaces.
If $A$ (resp. B) is a linear map on $\mathcal{H}_n$ (resp. $\mathcal{H}_m$),
then the linear map on the joint space $A\oplus B$ is given by the block diagonal matrix
with $A$ and $B$ on its diagonal.
Finally, $\mathcal{H}_n \otimes \mathcal{H}_m$ denotes the tensor product of the
two Hilbert spaces, and matrices on the joint space are obtained using the Kronecker product.

\section{Background}

\begin{figure*}[t]
\begin{align*}
  \Cx = \begin{bmatrix}
    1 & 0 & 0 & 0 \\ 0 & 1 & 0 & 0 \\ 0&0&0&1\\0&0&1&0 
  \end{bmatrix} &&%
  \Ry(\theta) = \begin{bmatrix}
    \cos(\sfrac\theta2) & -\sin(\sfrac\theta2) \\ \sin(\sfrac\theta2) & \cos(\sfrac\theta2)
  \end{bmatrix} &&%
  \Rz(\theta) = \begin{bmatrix}
    \exp(-i\frac\theta2) & 0 \\ 0 & \exp(i\frac\theta2)
  \end{bmatrix}
\end{align*}
\caption{Summary of gates and their associated matrices.}
\label{fig:gatemat}
\end{figure*}

We summarize here the salient points of the circuit model that are directly relevant to our work.
For readers wholly unfamiliar with the concepts, we recommend referring to standard
references such as~\cite{Nielsen2002} or the extensive Qiskit~\cite{Qiskit} documentation online%
\footnote{\url{https://qiskit.org/learn}}, which will provide better introductions.

In the circuit model, the basic unit of computation, the two-dimensional qubit, is represented
by a horizontal wire.
Reading from left to right, gates can then be applied to qubits in the form of boxes.
Boxes can act on a single-qubit, in which case they represent a $2\times 2$ unitary matrix,
or on $n$ qubits simultaneously -- with an associated $2^n\times 2^n$ matrix, the size of the matrix
growing exponentially with system size.
A typical example then looks like this
\begin{center}
  \includegraphics[width=.9\columnwidth]{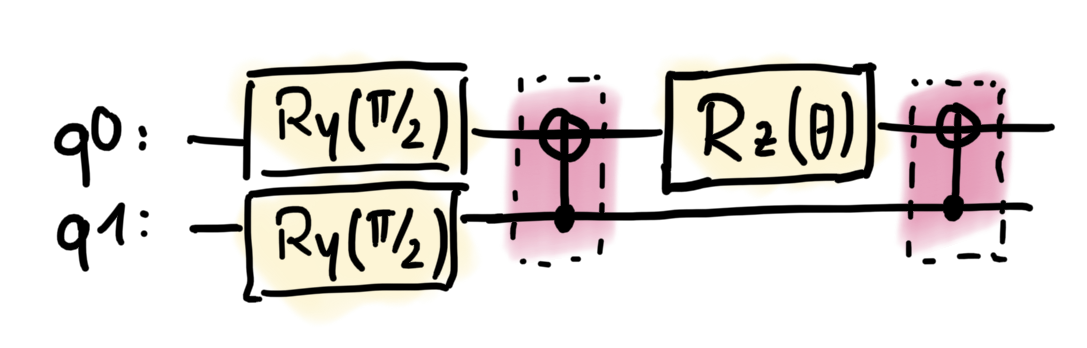}
\end{center}
In yellow are single-qubit gates (Y and Z-rotations)
the purple highlights two-qubit entangling \Cx{} gates.

To each qubit is associated a two-dimensional Hilbert space, its state space.
Crucially, as qubits are composed and considered as a joint system, their state space is
obtained by the tensor product $\mathcal{H}_1 \otimes \mathcal{H}_2$ of the individual
state space $\mathcal{H}_1, \mathcal{H}_2$.
This is the ``cause'' of the exponential size of unitary matrices.

A one-qubit rotation $\Ry(\theta) =${\ \footnotesize$
  \setlength{\arraycolsep}{2pt}
\begin{bmatrix}
  \cos(\theta) & \sin(\theta) \\ -\sin(\theta) & \cos(\theta)\\
\end{bmatrix}$}, applied to the second qubit of a two-qubit system, would become
\[
  \begingroup
  \setlength{\arraycolsep}{0pt}
  \Ry(\theta) \otimes \mathbb1_2 = \begin{bmatrix}
    \cos(\theta)  & \!\!      0       & \!\!\sin(\theta) & \!\!0\\
    0             & \!\!\cos(\theta)  & \!\!0            & \!\!\sin(\theta)\\
    -\sin(\theta) & \!\!   0          & \!\!\cos(\theta) & \!\!0\\
    0             & \!\!-\sin(\theta) & \!\!0            & \!\!\cos(\theta)\\
  \end{bmatrix}.
  \endgroup
\]
A remarkable result of quantum information theory is that arbitrary quantum computations
can be expressed in the circuit formalism from simple gate sets~\cite{Nielsen2002}.

A popular such \emph{universal} gate set is given by the entangling two-qubit \Cx{} gate
along with arbitrary single-qubit operations.
There are a variety of option to further
decompose single-qubit gates.
For our purposes, we will always consider the $\{\Ry, \Rz\}$ gate set, composed of Y-rotations
and Z-rotations.
The matrix representations of the \Cx, \Ry{} and \Rz{} gates is provided for reference
in \cref{fig:gatemat}.

A class of circuits of particular interest to us are reversible classical circuits~\cite{Shende2002}.
\begin{definition}
  A circuit whose matrix is given by a permutation $P\in S_{2^n}$ is called a \ul{reversible classical circuit}.
\end{definition}

The name refers to the fact that such circuits correspond to the subset of quantum circuits
that can be implemented in classical logic.
Such circuits can always be generated from the classical $\{\X, \Cx, \Ccx\}$ gate set, which correspond
to a NOT, a controlled-NOT and a doubly controlled NOT, also known as a Toffoli gate.

\section{Additive Quantum Circuits}
From now on, we refer to the commonly used circuit formalism summarized in the previous
section as the multiplicative, or $\otimes$-based model of quantum computation.
The phrasing highlights the structure of joint systems in that formulation.

In this section we present a new additive, or $\oplus$-based, quantum formalism.
We translate the quantum circuit primitives we introduced for the multiplicative language
to additive primitives and explore how they are related.

The next section will develop further the properties of this computation model, 
which will then be used in \cref{sec:circuit} to define a diagrammatic circuit-like
representation of additive computation.
In \cref{sec:examples}, we walk the reader through an example that illustrates the strengths
of the model.
We then present in \cref{sec:synthesis} a synthesis strategy that can convert any $\oplus$-based computation
into a standard $\otimes$-circuit, before concluding with a discussion in \cref{sec:discussion}.

\bingbat

Our definition of the additive model is based on the same $\{\Cx, \Ry, \Rz\}$ gate set introduced above.
The key difference is in how these gates, or rather their associated matrix, is modified when
they are considered within a larger-dimensional system.

We start with the simplest case, the \Rz{} rotation, and then proceed to the \Ry{} rotation, concluding
with the \Cx{} gate.

\subsubsection{Rz rotations}
\Rz{} rotations are one-qubit operations, and as such act non-trivially on a two-dimensional state space.
However, they can be rewritten up to global phase as
\begin{equation}
    \label{eq:phasegate}
    \Rz(\theta) = e^{-i\sfrac\theta2} \begin{bmatrix}
        1 & 0 \\ 0 & \exp(i\theta)
    \end{bmatrix}
\end{equation}
Global phases cannot be observed physically and as such the phase gate given by the matrix on the right-hand side
of \cref{eq:phasegate} are equivalent to the original \Rz{} gate.

We adopt this point of view in our additive formalism and define the additive $\Rz^+$ gate as a one-dimensional
operation.
\begin{definition}
    The \ul{additive Z-rotation $\Rz^+(\theta)$} is given by $x \mapsto \exp(i\theta)$ when $x\in\mathbb{C}^1$ is
    a one-dimensional state vector.
\end{definition}
This is generalized to arbitrary $n$-dimensional spaces using identities and the $\oplus$-structure:
\[\Rz^+_k(\theta) = \mathbb{1}_{k-1} \oplus \Rz^+(\theta) \oplus \mathbb{1}_{n-k}\]
and thus
\[
    \Rz^+_k(\theta)%
    \begin{bmatrix}
        x_1\\\vdots\\x_k\\\vdots\\x_n
    \end{bmatrix} = \begin{bmatrix}
        x_1\\\vdots\\\exp(i\theta)x_k\\\vdots\\x_n
    \end{bmatrix}.
    \label{eq:rzmulti}
\]
This corresponds to a diagonal matrix with diagonal $\big(1, \dots, 1, \exp(i\theta), 1, \dots, 1\big)$.

\subsubsection{Ry rotations}
\Ry{} rotations are one-qubit operations too, but unlike their $Z$ counterparts, they cannot be expressed
as one-dimensional operations in the computational basis.
Their definition is thus directly taken over from the multiplicative definition for the case of two-dimensional
systems.
\begin{definition}
    The \ul{additive Y-rotation $\Ry^+(\theta)$} is given by $\vec{x} \mapsto \Ry(\theta)\vec{x}$ when $\vec{x}\in\mathbb{C}^2$ is
    a two-dimensional state vector.
\end{definition}
When applied on a system with $n \geqslant 2$ dimensions, we specify the basis vector indices $(i,j)$ as subscripts.
For instance for $(i, j) = (1, 2)$ we have
\[
    \Ry^+_{1, 2}(\theta) \begin{bmatrix}
        x_1 \\ x_2 \\ \vdots \\ x_n
    \end{bmatrix} = \left(\Ry(\theta) \oplus \mathbb{1}_{n-2}\right) \begin{bmatrix}
        x_1 \\ x_2 \\ \vdots \\ x_n
    \end{bmatrix}.
\]
Note how these $n$-dimensional matrices look very different from their multiplicative counterparts.
For a four-dimensional, 2-qubit, state space, we would have
\begin{equation}
    \mathbb{1}_2 \otimes \Ry(\theta) = \Ry(\theta) \oplus \Ry(\theta) = \begin{bmatrix}
        \Ry(\theta) & \mathbb{0}_2 \\ \mathbb{0}_2 & \Ry(\theta)
    \end{bmatrix}
    \label{eq:multadd}
\end{equation}
which is very different from the additive
\[
    \mathbb{1}_2 \oplus \Ry(\theta) = \begin{bmatrix}
        \mathbb{1}_2 & \mathbb{0}_2 \\ \mathbb{0}_2 & \Ry(\theta)
    \end{bmatrix}.
\]
The difference becomes even more marked with larger systems.
The $n$-dimensional state space considered within the additive formulation
will have nothing to do with the $2^n$ dimensional state space of $n$ qubits.

\subsubsection{CX gates}
The \Cx{} gate is a two-qubit, and hence four-dimensional, operation.
Yet looking at its matrix representation in \cref{fig:gatemat}
we see that it only acts non-trivially on two of the four dimensions,
where it acts as an $\X =${\ \footnotesize$%
  \setlength{\arraycolsep}{2pt}%
  \begin{bmatrix}0 & 1\\ 1 & 0\end{bmatrix}$}
gate.

In the additive model, the four-dimensional \Cx{} gate can thus equivalently
be seen as the two-dimensional \X{} gate.

\begin{definition}
    The \ul{additive gate $\X^+$} is given by $\vec{x} \mapsto X\vec{x}$ when $\vec{x}\in\mathbb{C}^2$ is
    a two-dimensional state vector.
\end{definition}

The $\X^+$ gate is in fact just a swap of two dimensions;
when acting on two arbitrary coordinates of a larger space,
these swaps are the transpositions that can generate arbitrary permutations of the dimensions.
We thus get as a direct consequence:
\begin{corollary}
    The set of all circuits generated from $\X^+$ gates on $2^n$-dimensional state space, for $n \in \mathbb{N}$
    are precisely the reversible classical circuits on $n$ qubits.
\end{corollary}%
Remark that because of this, it actually suffices to consider additive primitves acting on nearest neighbors.
Nearest neighbor $\X^+$ transpositions generate arbitrary permutations, which can in turn be used
to rearrange wires so that $\Ry^+$ rotations always apply to neighbors.
We will make use of this fact for convenience in the diagrammatic representation of \cref{sec:circuit}.

\bingbat

More generally, we can summarize the exposition in this section with the following result.
\begin{thm}
    \label{thm:universality}
    The additive gate set $\{\Ry^+, \Rz^+, \X^+\}$ is universal, in the sense that any computation on $n$ qubits
    can be expressed with this gate set in $2^n$-dimensional space.
\end{thm}
\begin{proof}
    We need to show that any \Ry{}, \Rz{} and \Cx{} gate in a multiplicative $n$-qubit circuit can be expressed as additive
    gates in $2^n$ dimensional space.
    The result then follows from universality of the gate set.

    Consider a $k$-qubit multiplicative operation $A_k$ within an $n$-qubit system.
    As qubits can be permuted using classical reversible circuits,
    we can assume without loss of generality that $A_k$ acts on the first $k$ consecutive qubits of the system.
    Generalizing \cref{eq:multadd}, we observe that the $2^n\times2^n$ matrix
    $\mathbb{1}_2 \otimes \cdots \otimes \mathbb{1}_2 \otimes A_k = \mathbb{1}_{2^{n-k}} \otimes A_k$
   of the joint system
    is a block-wise diagonal matrix
    \begin{equation}
        \label{eq:distr}    
        \mathbb{1}_{2^{n-k}} \otimes A_k = \bigoplus_{i=1}^{2^{n-k}} A_k.
    \end{equation}
    The expression on the right-hand side corresponds to an additive operation with
    $2^{n-k}$ gates in parallel.

    Plugging $k = 1$ and $A_1 = \Ry(\theta)$ in \cref{eq:distr} gives us an equivalent additive expression
    for the gate $\Ry$: 
    \[\mathbb{1}_{2^{n-k}} \otimes \Ry(\theta) = \bigoplus_{i=1}^{2^{n-k}}\Ry^+(\theta).\]
    Similarly, with $k=1$ and $A_1 = \Rz(\theta)$, we use \cref{eq:phasegate} to express the $\Rz$ gate
    in the $\oplus$-language
    \[
        \Rz(\theta) = e^{-i\sfrac\theta2}\left(\mathbb{1}_1 \oplus \Rz^+(\theta)\right).
    \]
    Finally, for the case \Cx{} we use $k = 2$ and $\Cx{} = \mathbb{1}_2 \oplus \X^+$.
\end{proof}

\subsubsection{$\otimes$, $\oplus$ and the distributivity law}
There is a reason why we insist on calling the usual $\otimes$-based computations
multiplicative and our $\oplus$-based additive.
Note how the simple arithmetic formula
\[
    2 \times \cdots \times 2 \times a = 2^{n-k} \times a = \sum_{i=1}^{2^{n-k}} a
\]
is syntactically identical to \cref{eq:distr};
in fact $\mathbb{C}^n$ with operations $(\otimes, \oplus)$, zero $\varnothing$ and one $\mathbb{1}_1$
is a semiring isomorphic to $(\mathbb{N}, +, 0, \times, 1)$, under the map $\mathbb{C}^n \mapsto n$.
It is thus justified to view $\oplus$ as vector space addition
and $\otimes$ as vector space multiplication.
This has been studied and formalized very elegantly~\cite{Lack2004}.

In practice, this means we can not only consider $\otimes$-based and $\oplus$-based computations
separately, but that they can be combined into a single unified theory with well-defined
formal semantics.
In this paper, we focus on introducing $\oplus$-based computations and leave
the exploration of a joint framework for future work.

\section{Intuitions for Additions}
\subsubsection{Unitary Gaussian Elimination}
The additive model just presented is attractive for its remarkable simplicity.
In fact, the primitives introduced are 
the unitary equivalent of the elementary operations in Gaussian elimination.

Indeed, the first elementary matrix, a row swap, corresponds directly to the $\X^+$ operation.
The further two linear operations of Gaussian elimination,
row multiplication and summation of one row onto another,
are replaced by unitary equivalents:
multiplication is unitary as long as the scalar is a normalized complex number ($\Rz^+$ operation),
while row summation is replaced by a rotation that mixes two rows unitarily ($\Ry^+$ operation):
\[
    \begin{bmatrix}
        x \\ y
    \end{bmatrix} \xmapsto{\Ry}
    \begin{bmatrix}
        tx - \sqrt{1 - t^2}y \\
        ty + \sqrt{1 - t^2}x
    \end{bmatrix}
\]
where we rewrote the $\Ry$ definition using the mixing parameter $t = \cos(\theta)$.

Going one step further, we can give each unitary elementary matrix a clear interpretation
in terms of simple rotations on real two-dimensional subspaces.
This allows us to formulate an interesting intuitive characterisation of unitary transformations.

\subsubsection{Rotations in two-dimensional space}
Additive primitives within a joint system are always the identity on dimensions
outside of the specifically targeted space;
they are local operations.
This is in stark contrast to the usual $\otimes$-primitives, which
from the additive perspective, always act globally on the entire vector space.

This allows us to consider the action of $\Ry^+$, $\Rz^+$ and $\X^+$ locally, independently of system size.
The largest system size that we need to consider is the two-dimensional complex Hilbert space $\mathbb{C}^2$.
The $\Ry^+$ rotation is precisely the rotation of this two-dimensional subspace
\[
\begin{bmatrix}
    a\\ b
\end{bmatrix}
\xmapsto{\Ry^+(\theta)}
\begin{bmatrix}
    \cos(\sfrac\theta2) & -\sin(\sfrac\theta2)\\
    \sin(\sfrac\theta2) & \cos(\sfrac\theta2)\\
\end{bmatrix}
\begin{bmatrix}
    a\\ b
\end{bmatrix}
\]
To make it more real, we can equivalently see it as a simultaneous rotation
in the real and imaginary parts of both vector components.
\[
\Ry^+(\theta)
\begin{bmatrix}
    x + iy\\ z + iw
\end{bmatrix}
= \Ry^+(\theta)
\begin{bmatrix}
    x \\ z
\end{bmatrix}
+ i\Ry^+(\theta)
\begin{bmatrix}
    z \\ w
\end{bmatrix}.
\]
The $\Rz^+$ rotation is its ``orthogonal'' rotation: instead of rotating both real
and imaginary parts of two dimensions,
it rotates the real with the imaginary part within each dimension:
\begin{align*}
    &\ \Rz^+(\theta)(x+iy) = e^{i\theta}(x+iy)\\
    =& \left(\cos(\theta)x - \sin(\theta)y\right) + i\left(\sin(\theta)x + \cos(\theta)y\right).
\end{align*}

\subsubsection{Completing the picture: reversible clasiscal circuits}
The only element missing is the plumbing: the ability to permute dimensions as needed.
This is precisely what is given by the transposition $\X^+$.

Such permutations can always be commuted through to one end of the circuit:
given a permutation $\sigma$ of the wires and its associated permutation matrix $S$:
\begin{itemize}
    \item[--] an operation $\Rz^+_i$ will transform as \[S \circ \Rz^+_i = \Rz^+_{\sigma(i)} \circ S,\]
    \item[--] an operation $\Ry^+_{i,j}$ will transform as \[S \circ \Ry^+_{i,j} = \Ry^+_{\sigma(i), \sigma(j)} \circ S.\]
\end{itemize}

This allows us to entirely ignore the ordering of the dimensions in our representation,
as permutations of the dimensions can be introduced at will.
These clear intuitive semantics 
together with the formal properties of the theory will be the foundation for a succinct
circuit-like diagrammatic representation
that we will describe in the next section.

\section{Additivity as a Circuit}
\label{sec:circuit}
\begin{figure*}
  \includegraphics[width=\textwidth]{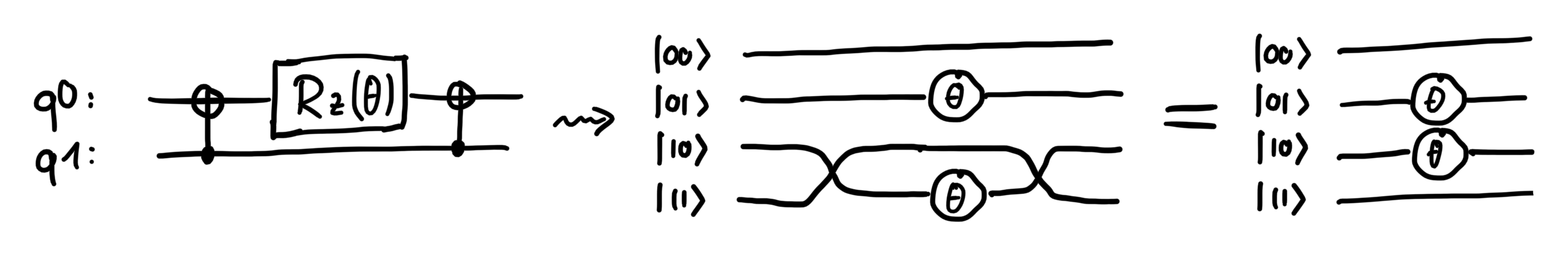}
  \caption{A phase gadget (i.e. a diagonal unitary operation) in the $\otimes$-representation (left) and
  its equivalent as an $\oplus$-circuit, with wire reordering simplification.}
  \label{fig:phasegadget}
\end{figure*}

Syntactically, additive circuits look very similar to the familiar multiplicative kind:
wires run from left to right and boxes can be appended on wires to represent linear maps
and composition.

The crucial difference that is worth highlighting once more is that unlike the usual quantum circuits
where each wire represents a qubit, in our representation each wire is a dimension of the underlying
state space.
A 3-qubit circuit is thus represented by 8 wires in the additive model.

Unlike in $\otimes$-circuits, wires are free to cross each other as they like
to connect to different gates.
For that we introduce the wire swap operation $\tau$
\begin{center}
\includegraphics[width=0.4\columnwidth]{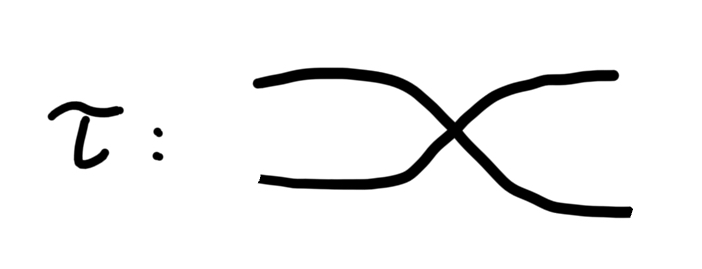}
\end{center}
that corresponds to the $\X^+$ gate -- a transposition that swaps two dimensions.
Applying a swap twice is again the identity:
\begin{center}
\includegraphics[width=0.8\columnwidth]{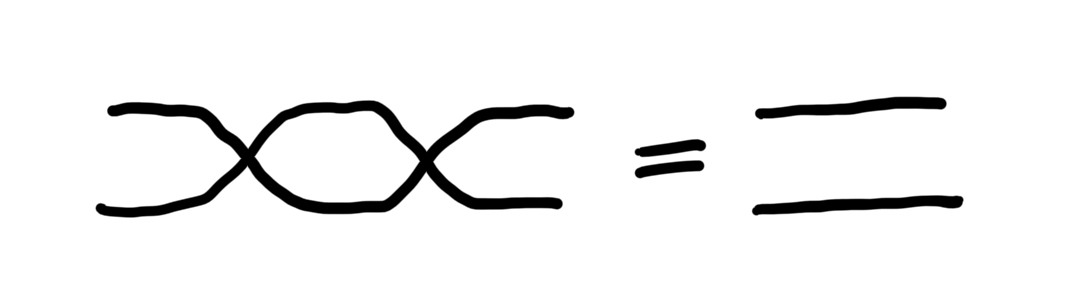}
\end{center}
Composing such swaps on different wires, we can express arbitrary permutations.
The precise swap order is immaterial within a given permutation.
In theory we would be free to write our 3-qubit, 8-dimensional identity circuit
as
\begin{center}
\includegraphics[width=\columnwidth]{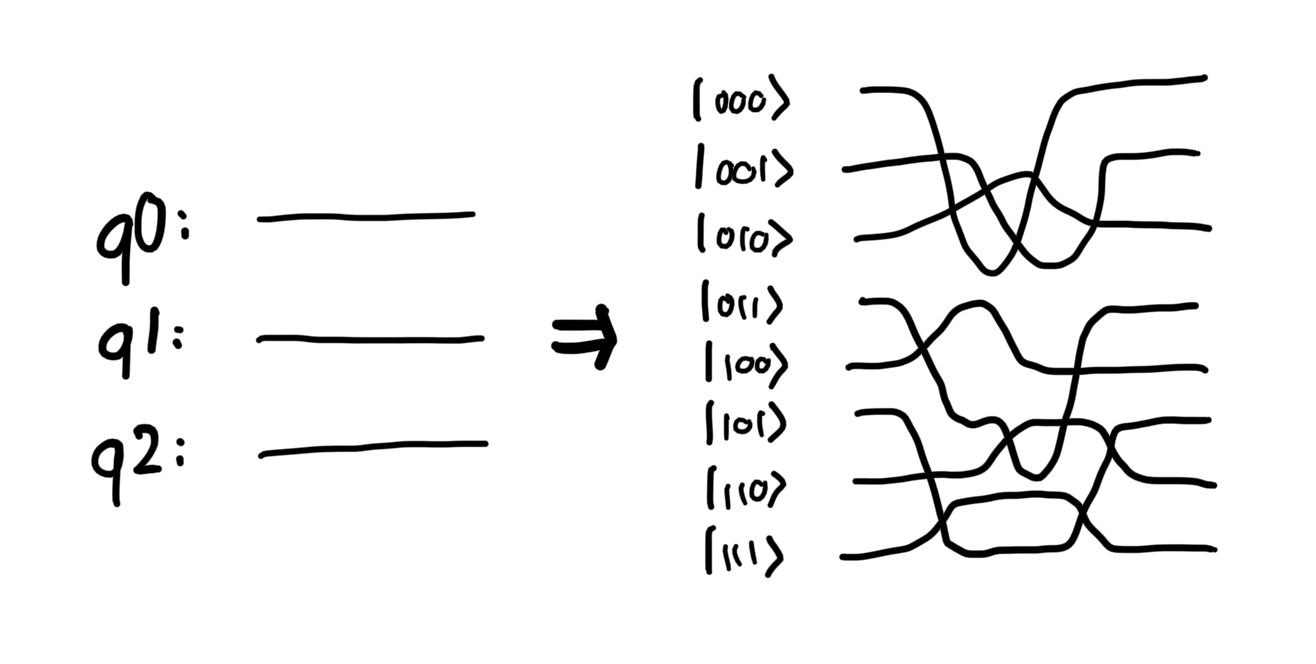}
\end{center}
This is an identity as long as the ordering of the input wires on the left matches the output
wires on the right.
There is no point doing this, however, so we will stay away from such obstruse drawings.

The point made is more general: the exact wire ordering and swaps in additive circuits is immaterial.
It can be left unspecified and for example used as degrees of freedoms by a compiler for smart optimizations.
That a wire ordering is fixed when circuits are drawn is merely an artifact of the two-dimensional representation.
We will see in \cref{sec:synthesis} that we can also represent additive circuits as graphs, in which case
no wire ordering is given.

\subsubsection{Rotations}
We represent the $\Rz^+$ rotation, acting on a single wire,
by a circle parametrized by $\theta$. 
\begin{center}
  \includegraphics[width=0.6\columnwidth]{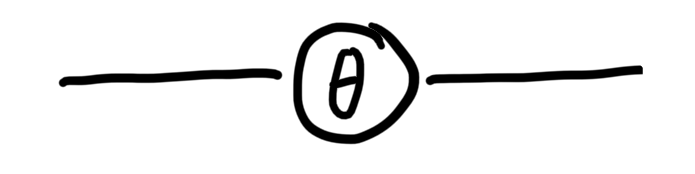}
\end{center}
A single $\Rz^+$ rotation on a system with two wires,
is equivalent to a single qubit \Rz{} rotation.
up to global phase:
\begin{center}
  \includegraphics[width=\columnwidth]{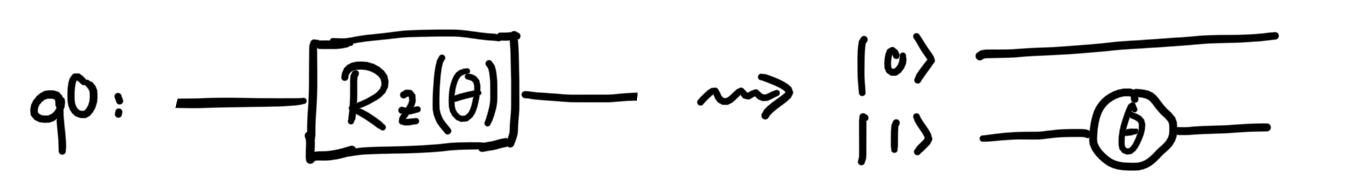}
\end{center}
A two-dimensional $\Ry^+$ rotation on the other hand
is represented by a box over two wires, parametrized
by $\theta$.
\begin{center}
    \includegraphics[width=\columnwidth]{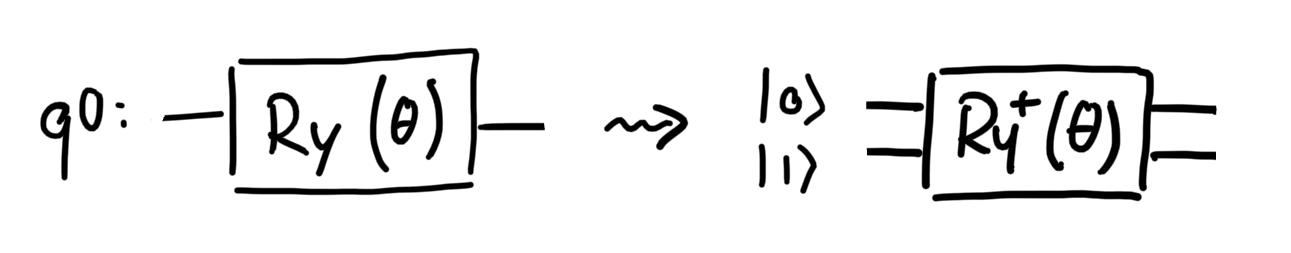}
\end{center}

\begin{figure*}
  \centering
  \includegraphics[width=\textwidth]{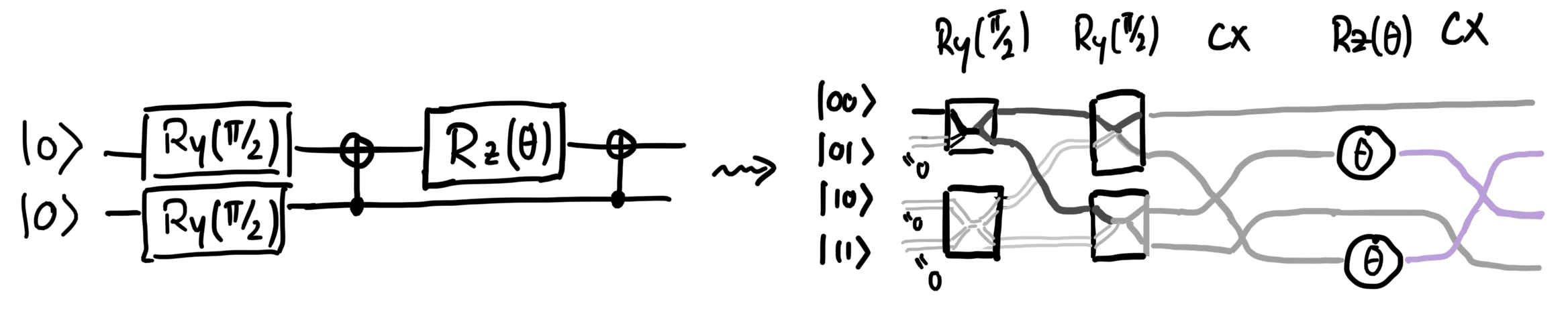}
  \caption{A computation from the $\ket{0}$ state, represented as a $\otimes$-circuit (left)
  and a $\oplus$-circuit (right).
  Wire transparency encodes amplitude -- transparent (amplitude $0$) to opaque (amplitude $1$).
  Color hue indicates phase -- the $\Rz(\theta)$ phase shift changes hue from black to purple.}
  \label{fig:acexample}
\end{figure*}
\subsubsection{Manipulating Additive Circuits}
Using the three primitives introduced ($\Ry^+$, $\Rz^+$ and $\X^+$),
we can represent any quantum computation as an additive circuit.
\Cref{fig:phasegadget} gives a simple example of how a $\otimes$-circuit
is translated into a $\oplus$-circuit.

Note how wire reordering is used to simplify the representation.
Such circuit manipulations are formalized by the following set of identities
\begin{thm}[Additive circuit identities]
    \label{thm:circid}
    Following identities can be used to transform additive circuits:
    \begin{enumerate}
        \item $\Ry^+$, $\Rz^+$ and $\X^+$ can be pushed through wire swaps $\X^+$
        \begin{center}
            \includegraphics[width=0.6\columnwidth]{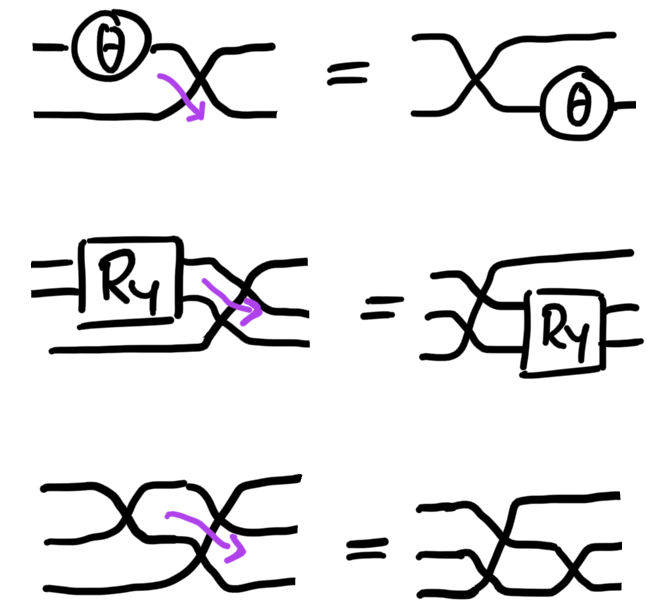}
        \end{center}
        \item Swapping the input wires of $\Ry^+(\theta)$ corresponds to a sign flip of the 
              rotation angle
        \begin{center}
            \includegraphics[width=.8\columnwidth]{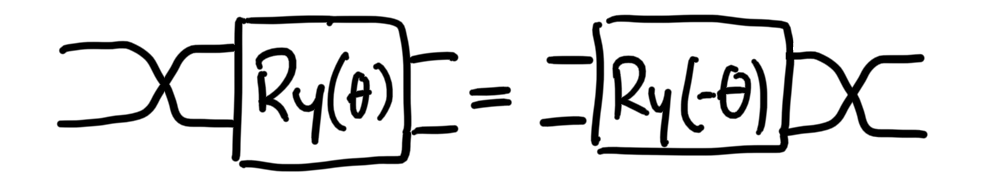}
        \end{center}
        \item Pairs of $\Ry^+$ or $\Rz^+$ gates can be merged when acting on the same wires
        \begin{center}
            \includegraphics[width=.9\columnwidth]{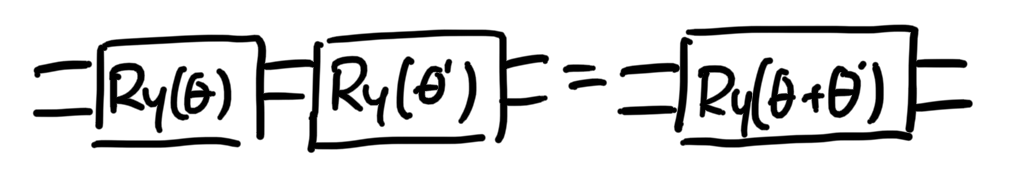}
        \end{center}
        \item Pairs of $\Ry^+$ or $\Rz^+$ gates that act on different wires commute.
        \begin{center}
            \includegraphics[width=.82\columnwidth]{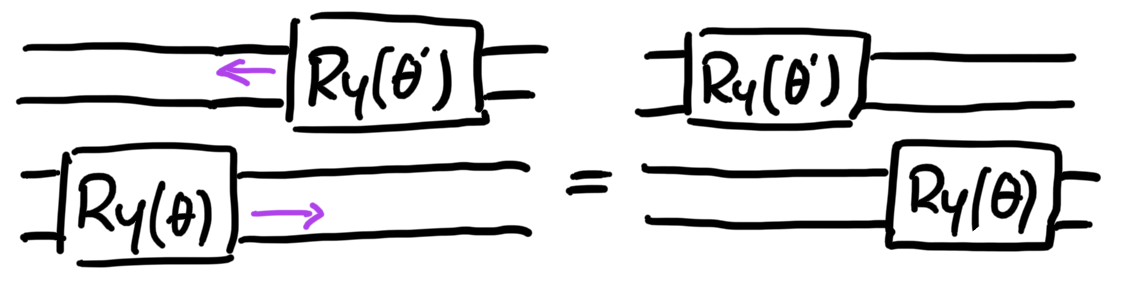}
        \end{center}
      \end{enumerate}
\end{thm}
The proof of these identities is very straightforward
but is given for completeness in \cref{sec:appendix}.
We will use these identities extensively in the examples \cref{sec:examples}.

\subsubsection{Viewing additive circuits as matrices}
\label{sec:matres}
One of the features of the additive circuit representation is that the evolution of the state vector
through the computation can be easily visualized.
We can achieve this by displaying the phase and amplitude of each dimension at any moment in the computation
on its corresponding wire in the $\oplus$-circuit.

In particular, by activating a single input bitstring -- that is, by setting the input to a computational
basis vector $\ket{b_1\cdots b_n}$, corresponding to one of the input wires --
we can compute its image $U\ket{b_1\cdots b_n}$ which corresponds to one of the column of the unitary $U$.
We can thus read off the matrix components from the additive circuit representation.

We are often particularly interested in the image $U\ket{0\cdots0}$ of the zero state, as this corresponds to the
computation typically executed on hardware.
\Cref{fig:acexample} shows how such a computation might be visualized as an additive circuit.

\section{An Example}
\label{sec:examples}
\begin{figure*}
  \includegraphics[width=\textwidth]{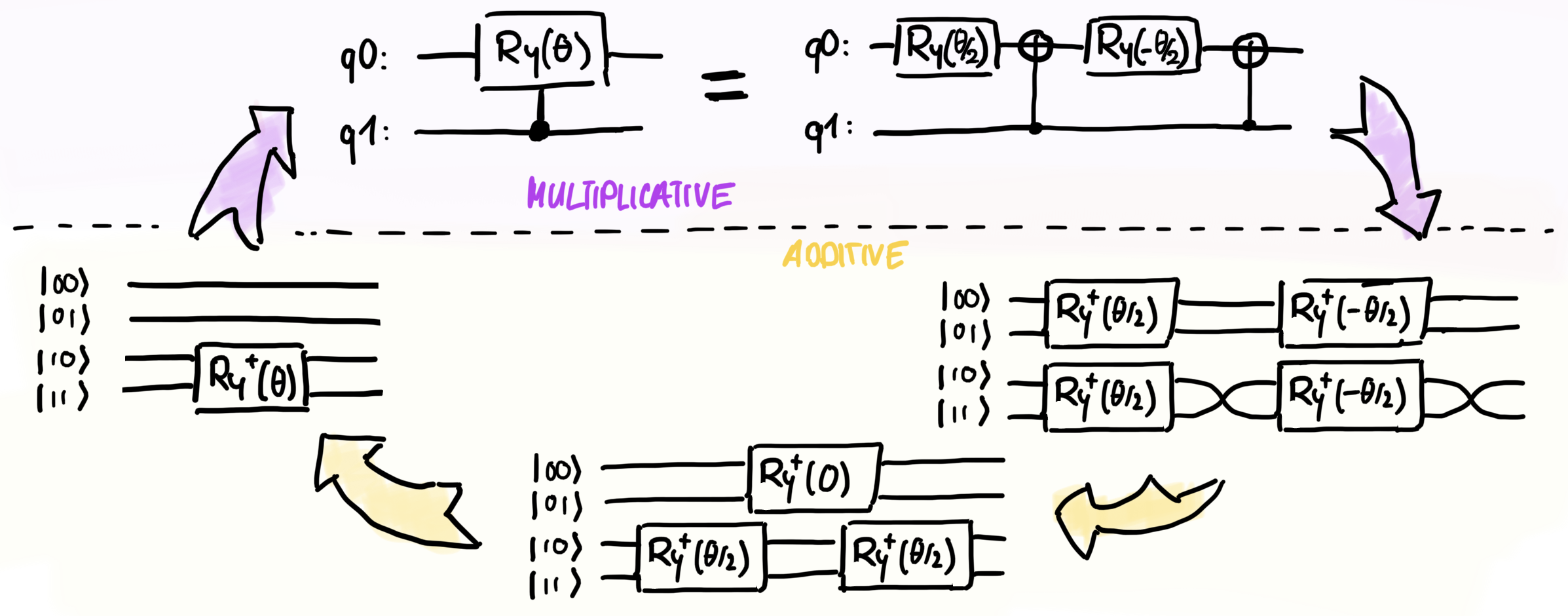}
  \caption{A controlled rotation, decomposed and recomposed using additive.}
  \label{fig:cry_ex}
\end{figure*}

Having simple, canonical representation of certain quantum operations makes it not only
very easy to understand what is happening, but it is also an ideal starting point
for optimization strategies.

There are a few types of circuits that the additive model can describe particularly well.
We already saw examples for two very common circuit structures, the reversible classical 
circuits that disappear into implicit wire reordering operations
and diagonal operations which appear as identities
with phase labels on some of the edges (\cref{fig:phasegadget}).

Another category of circuits that are suited to additive circuits are multi-controlled and multiplexed
operations~\cite{multiplexor}.
These are operations that consider a set of control qubits, on which they act as the identity,
and a set of target qubits that are transformed according to a unitary specified by the bitstring
of the control qubits.

To keep our exposition simple,
we will consider the example of a $\textsc{CRy}(\theta)$ gate acting on two qubits.
Gates with higher number of qubits and more complex controlled patterns can be processed and 
represented just as elegantly in much the same way.

Consider \cref{fig:cry_ex}.
The \textsc{CRy} gate and its decomposition into the universal $\{\Ry, \Rz, \Cx\}$ gate set is
shown at the top of the figure.
Following the arrows clockwise, we first see how the $\otimes$-circuit in the universal gate set (top right)
is translated into a $\oplus$-circuit (right-most circuit).

Each original $\Ry$ rotation is split into two $\Ry^+$ rotations.
This allows to merge rotations
that could not be combined in the multiplicative representation.
We also untangle the wires by reordering the inputs of one of the $\Ry^+$ rotations
using \cref{thm:circid}.
This yields the circuit at the bottom of the figure.

Removing the zero-angle rotation that is an identity, we find the canonical additive circuit
we would expect for the \textsc{CRy} gate: it applies a $\Ry$ rotation if the first qubit is set
and is the identity otherwise.
From this, the synthesis algorithm presented in \cref{sec:synthesis}
is able to synthesize the original \textsc{CRy} gate.

We could also provide an alternative decomposition of the \textsc{CRy} gate, in which the
$\Ry(\sfrac\theta2)$ rotation on the first qubit happens after the second \Cx.
Such degrees of freedom are hard to visualize in the $\otimes$-circuit.

The additive representation of the \textsc{CRy} gate on the other hand makes this commutativity evident.
Without considering the simplifications and the synthesis of the $\textsc{CRy}$ gate,
the $\oplus$-circuit of the original decomposition on the right-hand side shows clearly that
the first layer of $\Ry^+$ rotations in the $\oplus$-circuit can be moved past the second layer.
Expressed in the multiplicative circuit:
\begin{center}
    \includegraphics[width=\columnwidth]{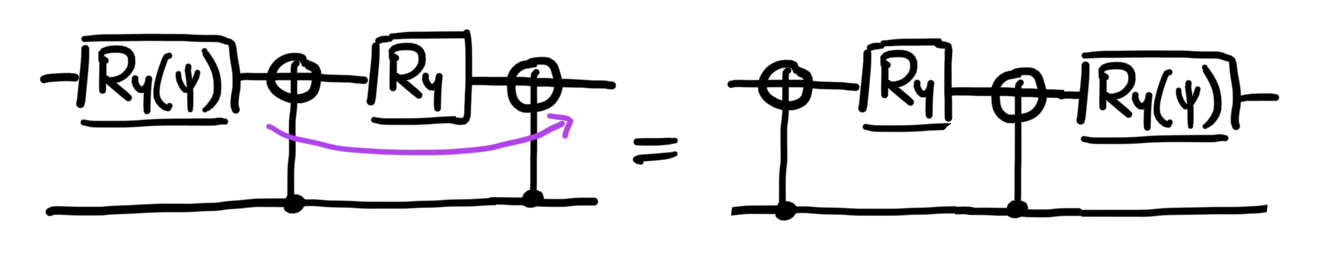}
\end{center}
Using alternative decompositions of the \textsc{CRy} gate yields the exact
same additive circuit.

\bingbat

That we are able to synthesize higher-level primitives such as the \textsc{CRy} gate -- the same techniques
would also succeed with multi-controlled gates, toffolis and more --
from its decomposed circuit is particularly powerful for optimizing compilers,
as we can hope to leverage the more abstract primitives for advanced optimizations.

This is also invaluable when compiling to architectures that support primitives that were not
in the original gate set.
Typical examples are multi-qubit gates such as Mølmer–Sørensen gates available
in ion traps~\cite{Grzesiak2020,Figgatt2019,Martinez2016} --
in the absence of such gates in the input circuit, current compilers are unable to leverage
such primitives.

To our knowledge, none of the current industrial quantum compilers~\cite{tket,cirq,Qiskit}
are able to synthesize such primitives from decomposed circuits, nor are we aware of any other
similar effort in the literature.

\section{Circuit Synthesis}
\label{sec:synthesis}
\begin{figure*}
  \includegraphics[width=\textwidth]{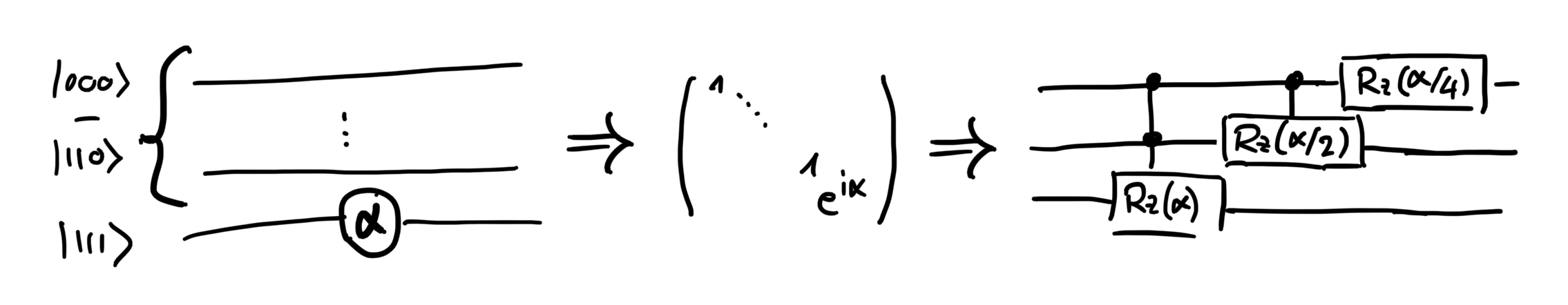}
  \caption{A single phase shift (left) translated as a multiplicative circuit (right).}
  \label{fig:phaseshift}
\end{figure*}

Previous sections have introduced the additive model of quantum computation
and its diagrammatic representation.
We have also shown in \cref{thm:universality} how arbitrary multiplicative circuits
can be translated into the additive model, and in \cref{sec:matres}
how the matrix representation can be read off from additive diagrams.

Before concluding, we would like to develop one more essential aspect of the theory,
namely the synthesis of a multiplicative circuit from the additive model.
We call this problem $\otimes$-circuit synthesis.

An efficient transformation between the two will be central to making
the additive formalism viable both for user-facing circuit design as well as circuit optimization,
so that computations formulated in the additive framework can be efficiently run
on physical hardware.

\subsubsection{Naive synthesis}
Each additive primitive can be expressed as a multiplicative circuit.
The simplest approach to synthesize a multiplicative circuit would
consider each additive primitive separately and compose their multiplicative
circuit representations one by one.

A general strategy can be devised. Consider a one or two-dimensional operation $A$
and suppose without loss of generality that it acts on the first one or two
dimensions of a larger system $A \oplus \mathbb{1}$\footnote{%
In practice implementing this permutation might incur a significant quantum cost.}.
In matrix representation, this corresponds to the matrix
of $A$ enlarged with additional rows and columns with ones on the diagonal and zeros everywhere
else.

From a multiplicative qubits perspective, this corresponds to applying operation $A$ on the first qubit
if and only if all other qubits are set: it is a controlled gate, controlled on all other qubits.
We immediately get that each $\Ry^+$ gate becomes a multi-controlled \Ctrln{\Ry^+} operation,
and a swap $\X^+$ gate becomes a multi-controlled \Ctrln{\X^+}, which in the special cases of 2 and 3 qubits
is a \Cx{}, respectively a Toffoli gate.

The $\Rz^+$ rotation can be treated in the same way, with one caveat.
We have defined $\Rz^+$ to correspond to the multiplicative $\Rz$ gate only up to global phase.
When controlling it, we obtain additional controlled phases.

\Cref{fig:phaseshift} shows the special case of the decomposition of the $\Rz^+$ gate
as a cascade of multi-controlled \Ctrln{\Rz^+} rotations. 
We could equivalently choose to represent it as a controlled phase gate and leave it up
to optimizing compilers to find the best decomposition.

The naive strategy however will perform poorly in most cases.
In the case of converting a multiplicative circuit to an additive one and back,
this strategy gives us an exponential blow-up in the number of gates.
Indeed, a $n$-qubit circuit with a single \Ry{} rotation gives a $2^n$-dimensional
additive circuit with $2^{n-1}$ $\Ry^+$ gates with the construction of \cref{thm:universality}.

If we decompose each of these $\Ry^+$ gates individually, we will obtain an exponential
number of $n-1$-controlled operations that will each require at least $\mathcal{O}(n)$ many
primitive operations.
We can drastically improve on this.

\begin{figure*}
    \centering
    \includegraphics[width=\textwidth]{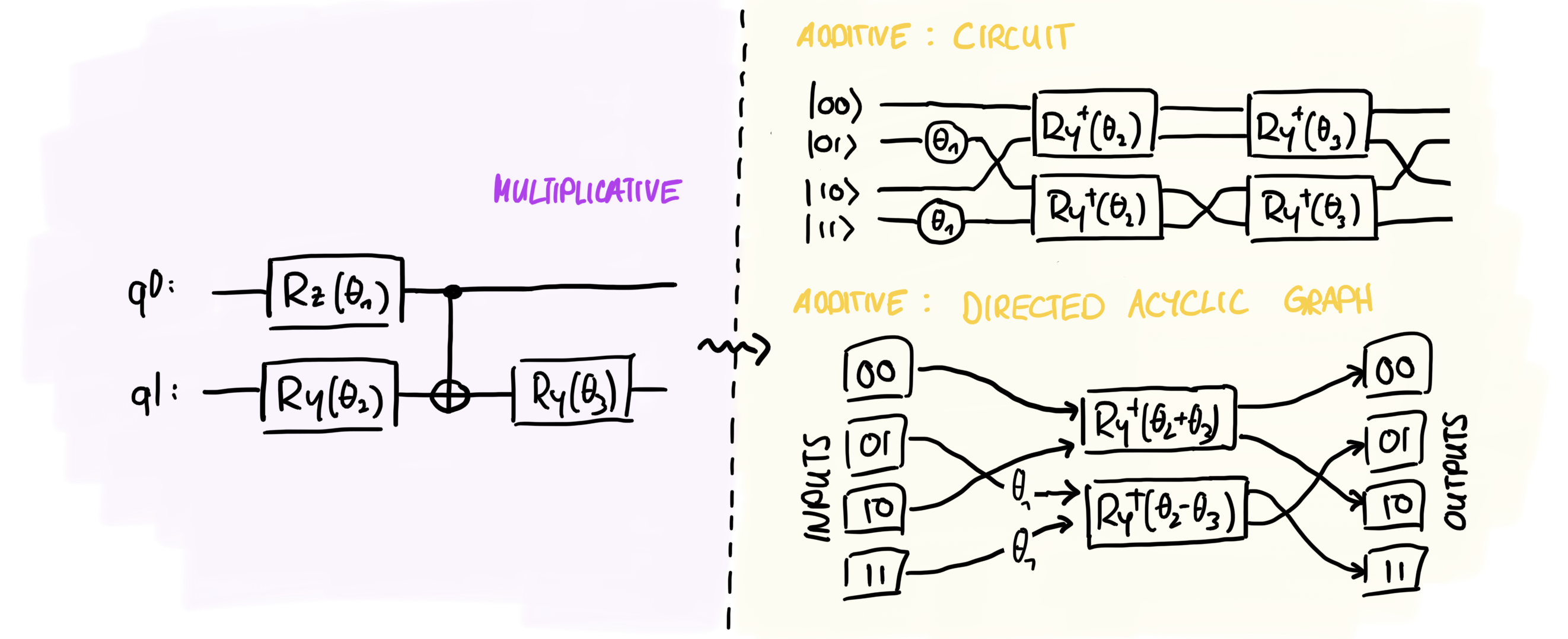}
    \caption{A quantum computation, expressed as a $\otimes$-circuit (left), a $\oplus$-circuit (top right) and an additive DAG (bottom right).
    Note that the additional wire swaps introduced in the $\oplus$-circuit for convenience of presentation are implicit in the DAG.}
    \label{fig:dag}
\end{figure*}
\subsubsection{Internal Representation}
We implement additive circuits 
as directed acyclic graphs (DAG) similar to how current quantum compilers represent circuits~\cite{tket,Qiskit}.
Additive circuits are particularly simple to store in DAGs and much of the redundancy in the visual
representation -- wire ordering and order of commuting operations -- is removed.
We will see that we can further abstract away $\Rz^+$ rotations, leaving only the connectivity
of $\Ry^+$ operations to be represented.

We introduce input and output vertices for each dimension of the vector space.
Each two-dimensional $\Ry^+$ operation becomes an internal vertex,
with two incoming and two outgoing edges, one for each dimension.

The incoming (outgoing) neighbors are the preceding (following) $\Ry^+$ gates that act on
the same dimension.
In the case of first and last operation on each dimension, the respective incoming and
outgoing edge connects to an input or output vertex.

We perform some additional transformations to store the DAG in a canonical representation.
We merge all consecutive $\Rz^+$ operations between two $\Ry^+$ gates into a single $\Rz^+(\theta)$ rotation.
The parameter $\theta$ of the rotation can then be added as a label to its corresponding edge.
Similarly, we merge any two neighboring $\Ry^+$ vertices if they act on the same pair
of dimensions.

To account for dimension permutation of $\Ry^+$ inputs, 
we use Identity 2 from \cref{thm:circid}
\begin{equation}
\begin{gathered}
\includegraphics[width=.9\columnwidth]{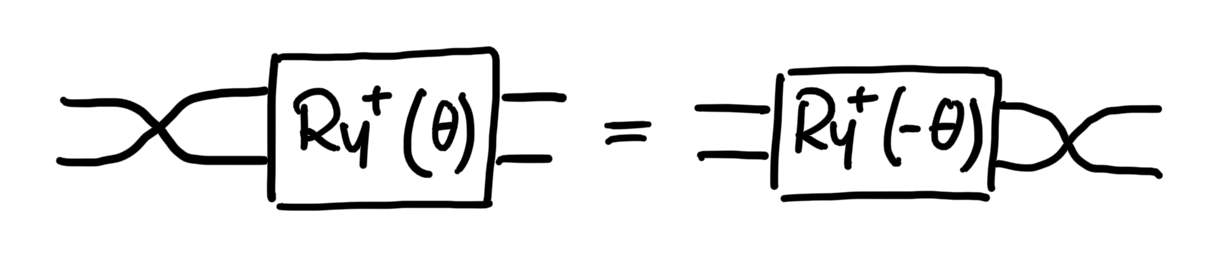}
\end{gathered}
\label{eq:cry_swap}
\end{equation}
so that the ordering of the edges at each vertex can always be set to a chosen fixed ordering.
An example of a circuit and its DAG representation is given in figure \cref{fig:dag}.

\subsubsection{Vertex fusion and stacked vertices}
A key step in the algorithm is identifying sets of $\Ry^+$ rotations that can be synthesized jointly:
While an isolated $\Ry^+$ rotation in a $2^n$-dimensional space will result in a $n-1$-controlled
rotation in multiplicative space, a joint set of $2^{n-1}$ $\Ry^+$ rotations can be synthesized
into a single-qubit \Ry{} rotation, using the construction of \cref{thm:universality},
by reordering the wires such that it can be written as
\[
\bigoplus_{i=1}^{2^{n-1}} \Ry^+(\theta) = \mathbb{1}_{2^{n-1}} \otimes \Ry(\theta).
\]

This can be generalized: a joint set of $2^{k-1}$ $\Ry^+$ rotations will correspond
to a $n-k$-controlled $\Ry$ rotation.

To support the grouping of vertices, we introduce a fusion operation between two vertices.
Two vertices with the same angle parameter can be fused if they are not causally related with each other --
that is, there is no directed path in the DAG from one vertex to the other --
and if their angles match.

\begin{figure*}
  \centering
  \includegraphics[width=\textwidth]{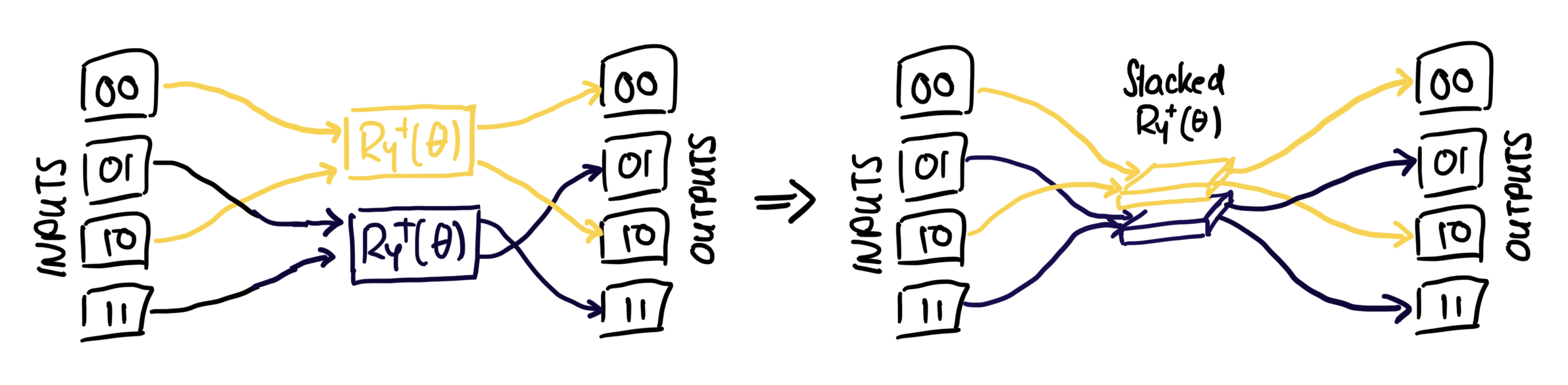}
  \caption{A simple additive circuit before (left) and after stacking (right).
  The resulting stacked vertex is of length 2.}
  \label{fig:stacked}
\end{figure*}

The fusion operation merges the incoming and outgoing edges of the two vertices into one
new vertex. This vertex must keep track of the incoming and outgoing edge pairs
that correspond to the pairs of edges
of the original unfused vertices.
\Cref{fig:stacked} illustrates vertex fusion.

With the picture in our minds of such a fused vertex, made up of a collection of unfused vertices,
we refer to them as stacked vertices.
The length of such a vertex is given by the number of unfused vertices that compose it.

The more rotations can be combined and synthesized together, the lower the
gate count of the resulting synthesized circuit.
However, synthesis is only possible for stacked vertices of length $2^k$, for some non-negative integer $k$.

We can achieve such power-of-two stacked vertices by decomposing each stack into an array of power-of-two stacks.
Given a stack of length $n$, this can be achieved in two ways.
The stack can either be split up into smaller power-of-two stacks using the binary decomposition of $n = 2^{k_1} + 2^{k_2} + \cdots$
or by enlarging the stack to the next power of two and introducing another smaller stack with inverse
operations on the additional dimensions, which can be decomposed recursively.

\subsubsection{Edge placement and routing}
Each dimension in the additive space must correspond to a bitstring in the multiplicative model.
This is subject to connectivity constraints, as controlled rotations can only be synthesized from power-of-two
stacked vertices when the rotation acts locally on one qubit.
A $\Ry^+$ operation acting on bitstrings \texttt{00} and \texttt{11}, for instance,
cannot be expressed as a local single-qubit rotation.

In analogy to the qubit routing problem~\cite{routingtket}, we refer to the map from dimensions to bitstrings
as the placement map.
The additive routing problem then consists in satisfying the connectivity constraints 
by modifying the placement map as we traverse the circuit, 
while minimizing the cost of the permutations introduced.

The connectivity constraints can best be understood on the hypercube graph, in which vertices are bitstrings and edges
exist between any two vertices that have Hamming distance one.
For three qubits, we get the three-dimensional cube
\begin{center}
  \includegraphics[width=0.5\columnwidth]{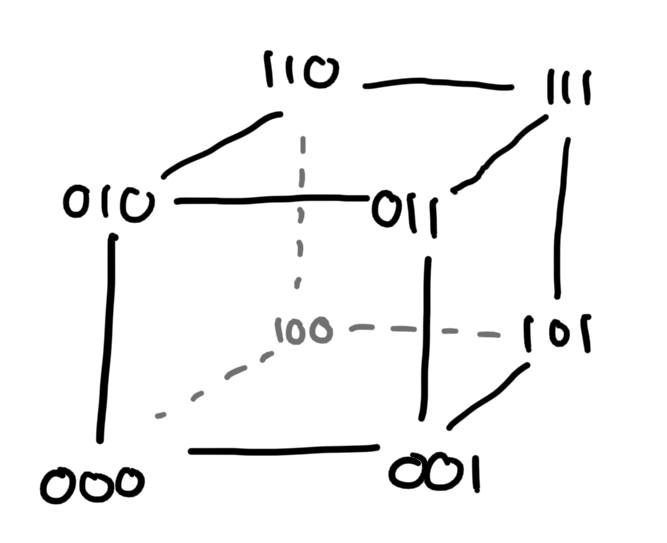}
\end{center}
The vertices and edges of the hypercube can be arranged in a grid in $n$-dimensional space
so that all parallel edges correspond to bit changes in the same bit.

\begin{figure*}
  \centering
  \includegraphics[width=.9\textwidth]{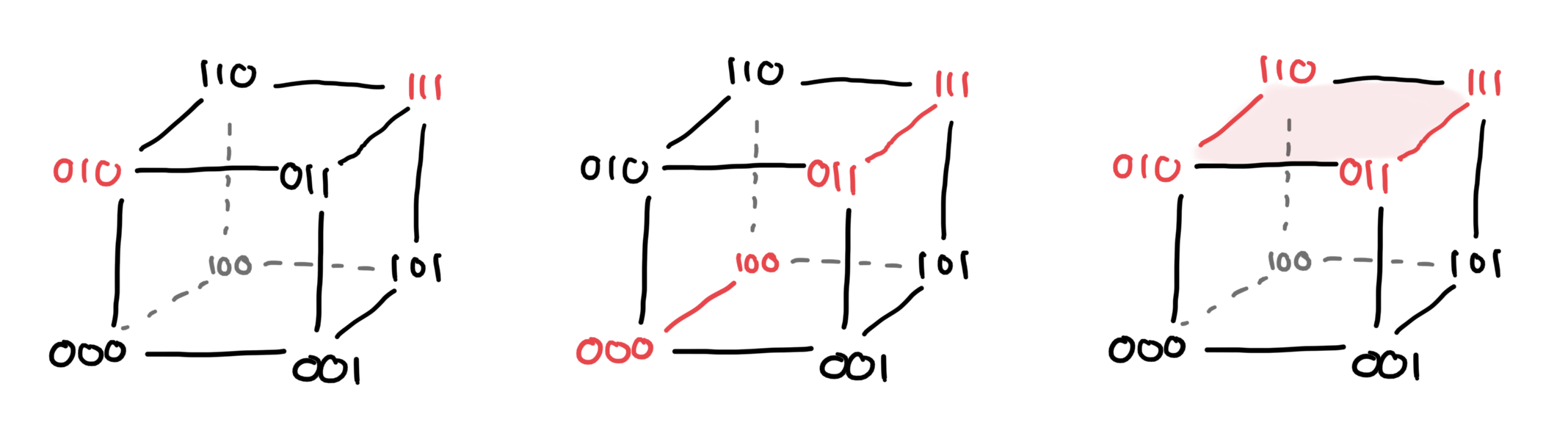}
  \caption{On the left, two bitstrings not fulfilling any constraint. In the middle, two pairs of bitstrings 
  satisfying the edge and parallel constraints, but not satisfying the dense constraint. On the right,
  two pairs of bitstrings satisfying all three constraints. This would be a valid placement for a stacked vertex.}
  \label{fig:constraints}
\end{figure*}

We can define vertex connectivity constraints as follows
\begin{definition}[Connectivity constraints]
    \begin{enumerate}\item[]
        \item 
  For any unstacked vertex, the two bitstrings on which it is placed must
  share an edge on the hypercube (\ul{edge constraint}).
        \item
  For any stacked vertex, each unstacked vertex of which it is composed must
  satisfy the edge constraint, and all resulting edges on the hypercube must
  be parallel (\ul{parallel constraint}).
        \item
  For any power-of-two stacked vertex of length $2^k$, all bitstrings on which it is placed must 
  live within a $k$-dimensional subspace of the $n$-hypercube (\ul{dense constraint}).
    \end{enumerate}

\end{definition}
In other words:
the edge constraints states that two bitstrings acted upon by an unstacked vertex
can only differ in one bit;
the parallel constraint requires that all pairs
of bitstrings from unstacked vertices within a stacked vertex
may only differ in the same bit;
and the dense constraint means any two bitstrings of a stacked vertex of length $2^k$
can only differ in at most $k$ bits.
\Cref{fig:constraints} illustrates these three constraints.

\begin{thm}
  Consider a power-of-two stacked vertex $v$ of length $2^k$ in an additive DAG and
  let $n > k$ be an integer number.
  If a placement map satisfies the edge constraint, the parallel constraint and the dense constraint with respect to $v$,
  then it can be factorized as a controlled $\textsc{C}^{n-k-1}\Ry$ rotation in a $n$-qubit circuit,
  up to conjugations with single-qubit \X{} gates on a subset of qubits.
  In the special case where $v$ has maximal length $k + 1 = n$, the controlled rotation is
  simplified to an uncontrolled $\Ry$ gate.
\end{thm}
\begin{proof}
  Using the edge and parallel constraints, there is a single qubit that is not left
  unchanged by the transformation. Assume it is the $n$-th qubit.
  Using the dense constraint, assume further that the first $n-k-1$ bits are identical across all bitstrings.
  We conjugate qubits with $\X$ gates where necessary, so that qubits $1$ to $n-k-1$ are set to 1 for all placed bitstrings.
  We thus get that the transformation associated with stacked vertex $v$ is given by
  \[
      \ket{b_1\cdots b_n} \xmapsto{v} \begin{cases}
          \ket{b_1\cdots b_{n-1}} \Ry(\theta)\ket{b_n}&\ \textrm{if }c = 1\\
          \ket{b_1\cdots b_n} &\ \textrm{otherwise.}
      \end{cases}
  \]
  where $c = \prod_{i=1}^{n-k-1} b_i$ and $\theta$ is the angle of the rotations in $v$.
  This is precisely a $\textsc{C}^{n-k-1}\Ry(\theta)$ gate controlled on qubits $1$ to $n -k -1$ with target qubit $n$.
\end{proof}

We can thus proceed with the initial placement and routing by traversing the vertices of the DAG
in topological order, permutating bitstrings whenever a vertex' connectivity constraints are not satisfied.
The overall cost of routing will greatly depend on choosing permutations of bitstrings -- that is, reversible classical circuits --
that can be implemented efficiently in a $\otimes$-circuit.

The established literature on reversible classical circuits provides a wide range of heuristics for optimized
circuit synthesis that could be leveraged for this routing problem~%
\cite{Shende2002,Susam20171,Li2014,Stojkovic2020,Scott2008231,Arabzadeh2010849,Datta2013209,Datta2013316,Gado2021}
A transposition-based approach as suggested in~\cite{Stojkovic2020} could be a fast first heuristic, as it is simple
to establish which bitstring swaps would solve the given connectivity constraints
from the hypercube model.

\subsubsection{Relative phase corrections}
We have so far focused on fusing, placing and synthesizing vertices in our DAG representation.
The last missing link is the implementation of the graph edges.
Phase information stored on the edges must be synthesized before the vertices can themselves be synthesized.

This can be readily achieved using well-studied phase polynomial synthesis algorithms
that show excellent performance in practice~\cite{Amy2019,Degriend2020,Vandaele2021}.

Note also that only relative phase differences between dimensions on a given vertex must be implemented;
this means that phase differences with other independent dimensions can be treated lazily,
implementing them only when they are required.
Combined with techniques to relax the reversible circuit synthesis problem to solutions up to relative
phase, this could prove a powerful optimization technique in the future~\cite{Maslov2016}.

\subsubsection{Synthesis Strategy in Summary}
We proposed a circuit synthesis algorithm that can convert a computation from its additive representation
into an executable multiplicative circuit.
The output quantum circuit can then be further optimized on existing quantum compilers~\cite{tket,Qiskit,cirq}
before execution on real hardware.
An implementation in Julia using the Yao framework~\cite{julia,YaoFramework} is in development
and will be made available soon.

In summary:
\begin{enumerate}
  \item Convert an additive circuit to its additive DAG representation.
In the canonical form, any consecutive $\Ry^+$ or $\Rz^+$ rotation will be merged.
  \item Stack vertices. Find vertices that can be factorized and synthesized jointly and fuse them together.
  \item Organize stacked vertices in stacks of size $2^k$.
  \item Satisfy connectivity constraints by placing and routing DAG edges on the hypercube.
  \item Convert to multiplicative circuit by in turn synthesizing controlled rotations from vertices,
  phase polynomials from edge labels and reversible classical circuits from placement map.
\end{enumerate}

\section{Conclusion}
\label{sec:discussion}
In this paper, we have introduced an additive model for quantum computation that
is a departure from all usual $\otimes$-based formalisms
such as the quantum circuit language and its derivatives.
In the additive formalism, widely used constructions such as controlled
unitaries, diagonal operations and reversible classical circuits
have a very intuitive presentation.

It is in its structure much closer to the matrix representation of unitary computations
but offers a graphical circuit-like representation that makes it
convenient to use for circuit design.
Unlike the matrix representation, additive circuits further maintain some circuit information,
allowing for circuit synthesis.
This removes the need for matrix decomposition techniques that are extremely
computationally expensive and scale poorly~\cite{Iten2019,Davis2020,Smith2021,Gheorghiu2021,Wu2020,Krol2021,Loke2014}.

Expressing quantum computations with $\oplus$-based primitives is also interesting for
circuit optimization purposes.
As an example, we saw that additive circuits can recover higher-level controlled operations
from its decomposition.
The example provided in \cref{sec:examples} was simple, but the same strategy easily
scales to larger primitives.

Decompositions of such multi-controlled operations quickly become complex,
making it a challenge for compiler internal representations to identify and leverage
such higher-level structure.
We are not aware of any tool able to perform this.
Additive circuits also make commutation relations of such complex structures apparent,
which cannot be inferred from local considerations in $\otimes$-based circuits.

One of the main challenges of the additive representation
is the $\otimes$-circuit synthesis probelem.
An efficient conversion from $\oplus$ to $\otimes$-circuits is essential
for widespread use of the technology.

Our proposal in \cref{sec:synthesis} makes first step towards a resolution.
We show that circuit synthesis strategies can leverage optimization strategies
from the widely successful fields of reversible classical circuits synthesis
and phase polynomial synthesis.
This gives us an important head start towards competitive performance.

The scaling of the additive representation in general will be exponential, so that it will
never be usable on large scale circuits.
It might still be competitive however with matrix synthesis, as finding efficient implementations
of unitaries is a particularly hard problem.

Promisingly, the distributivity rules of ($\otimes$, $\oplus$) mean that the two
circuit models could easily be combined and considered together.
Circuits could then be additive in parts where it provides an advantage, whilst
keeping the scaling of the $\otimes$-based model overall.
As an optimization technique, this would be closely related to block-wise matrix
synthesis, where a circuit is decomposed into smaller subcircuits that are optimized
individually using a matrix synthesis algorithm.
A circuit-based ($\otimes$, $\oplus$)-approach would be much more modular and composable, however.

Finally, the additive model provides an excellent opportunity to develop new intuitions on quantum computing
for educational and research purposes.
As a first step, we are developing software that will allow
users to design 
quantum circuits in the additive language
and to perform $\otimes$-circuit synthesis, so that arbitrary $\oplus$-circuits
could be compiled and optimized for execution on today's quantum hardware.

\section*{Acknowledgements}
My gratitude goes to Stefano, who helped me formulate many of the ideas I presented here.
I would also like to thank my excellent team and colleagues at Cambridge Quantum, and Silas
in particular for listening to my lengthy expositions.


{\small\printbibliography}

\appendix
\section{Proof of \Cref{thm:circid}}
\label{sec:appendix}
\begin{proof}
    Identity 1 is a special case of the simple equation
    \[SA = (SAS^\top)S\]
    where $S$ is some permutation matrix and $A$ a linear operation.
    In the case where we can write $A$ as $A = A_1 \oplus \cdots \oplus A_k$ and $S$
    leaves each $A_i$ subspace unchanged, then there is a permutation $\sigma$ of $1, \dots, k$
    such that
    \[
        S(A_1 \oplus \cdots \oplus A_k)S^\top = A_{\sigma(1)} \oplus \cdots \oplus A_{\sigma(k)}
    \]
    and thus
    \[
        S(A_1 \oplus \cdots \oplus A_k) = (A_{\sigma(1)} \oplus \cdots \oplus A_{\sigma(k)})S.
    \]
    Writing this graphically yields the statement.

    Identity 2 follows from the anticommutativity of the Pauli matrices $X$ and $Y$
    \begin{align*}
        X\Ry(\theta) &= X(\mathbb{1}_2 \cos(\theta) + iY\sin(\theta))\\
        &= (\mathbb{1}_2 \cos(\theta) - iY\sin(\theta))X = \Ry(-\theta)X.
    \end{align*}
    Identity 3 is a well-known property of rotations:
    $\Ry(\theta) + \Ry(\theta') = \Ry(\theta + \theta')$; the same holds for $\Rz$.

    Identity 4 can be written as
    \begin{align*}
        &\left(\mathbb{1}_2 \oplus \Ry^+(\theta) \right)\left(\Ry^+(\theta') \oplus \mathbb{1}_2\right)
        = \Ry^+(\theta') \oplus \Ry^+(\theta)\\
        = &\left(\Ry^+(\theta') \oplus \mathbb{1}_2 \right)\left(\mathbb{1}_2 \oplus \Ry^+(\theta)\right),
    \end{align*}
    and similarly for $\Ry^+$.
\end{proof}

\end{document}